%% file: paper.tex
\documentclass{llncs}

\usepackage{subfigure}
\usepackage{xspace}
\usepackage{rotating}
\usepackage{hhline}
\usepackage{amsmath}
\usepackage{amssymb}
\usepackage{paralist}
\usepackage{stmaryrd}
\usepackage{textcomp}
\usepackage{verbatim}
\usepackage{tikz}
\usepackage{graphicx}
\usepackage{pifont}
\usepackage{dsfont}
\usepackage{wrapfig}
\usepackage{listings}
\usepackage{cite}
\usepackage{pifont}
\usepackage{fancyhdr}

\usetikzlibrary{shapes}
\usetikzlibrary{calc}
\usetikzlibrary{arrows}
\usetikzlibrary{decorations.pathmorphing}
\usetikzlibrary{decorations.text}
\usetikzlibrary{plotmarks}

\input{macros}

\begin{document}

\pagestyle{headings}  

\title{Proving Safety with Trace Automata\\ and Bounded Model Checking}

\author{Daniel Kroening\inst{1} \and Matt Lewis\inst{1} \and Georg Weissenbacher\inst{2}}

\institute{University of Oxford \and Vienna University of Technology}

\maketitle  
\pagenumbering{arabic}

\input{abstract}
\input{introduction}
\input{motivation}
\input{preliminaries}

\input{reduce_diameter}

\input{safety_checking}
\input{experiments}
\input{related}

\input{conclusion}

\bibliographystyle{splncs}
\bibliography{paper}

\clearpage

\appendix

\newcommand{\refcounter}[2]{\setcounter{#1}{#2}\addtocounter{#1}{-1}} %
\input{appendix_proofs}
\input{results_table}

\end{document}

%% file: macros.tex
\newcommand{\emptypath}{\ensuremath{\varepsilon}}
\newcommand{\defn}{\ensuremath{\stackrel{\textup{\tiny def}}{=}}}

\newcommand{\CBMC}{{\sc Cbmc}\xspace}

\newcommand{\SVCOMP}{{SV-COMP14}\xspace}

\newcommand{\wlpcond}[2]{
  \ensuremath{wlp({#1},{#2})}\xspace
}

\newcommand{\true}{{\ensuremath{\mathsf{T}}}}
\newcommand{\false}{\ensuremath{\mathsf{F}}}

\newcommand{\Vars}{\ensuremath{\mathsf{Vars}}\xspace}
\newcommand{\Exprs}{\ensuremath{\mathsf{Exprs}}\xspace}
\newcommand{\BExprs}{\ensuremath{\mathsf{\mathds{B}\text{-}Exprs}}\xspace}
\newcommand{\Stmts}{\ensuremath{\mathsf{Stmts}}\xspace}
\newcommand{\PStmts}{\ensuremath{\mathsf{Stmts}_P}\xspace}
\newcommand{\AccPStmts}{\ensuremath{\mathsf{Stmts}_{\acc{P}}}\xspace}
\newcommand{\States}{\ensuremath{\mathsf{States}}\xspace}

\newcommand{\stmt}{\ensuremath{\mathsf{stmt}}\xspace}

\newcommand{\skipstmt}{\ensuremath{\mathsf{skip}}\xspace}
\newcommand{\trans}[3]{\ensuremath{{#1}\stackrel{{#3}}{\longrightarrow}{#2}}}
\newcommand{\sem}[1]{\ensuremath{\llbracket{#1}\rrbracket}}
\newcommand{\id}{\ensuremath{\mathsf{id}}\xspace}
\newcommand{\lang}{\ensuremath{\mathcal{L}}\xspace}

\newcommand{\state}{\ensuremath{\sigma}\xspace}
\newcommand{\acc}[1]{\ensuremath{\widehat{{#1}}}}
\newcommand{\bound}[1]{\ensuremath{\beta({#1})\xspace}}
\newcommand{\uacc}[1]{\ensuremath{\widetilde{{#1}}}}
\newcommand{\predovfl}[1]{\ensuremath{\varphi_{{#1}}}\xspace}
\newcommand{\stmtovfl}[1]{\ensuremath{\overline{\tau}_{{#1}}}\xspace}
\newcommand{\stmtnoovfl}[1]{\ensuremath{\tau_{{#1}}}\xspace}

\def\fork{
\raisebox{-1.15em}{\begin{tikzpicture}
  \node at (0,0cm) (first) [draw,circle,fill=white,text width=7pt,align=center,
        inner sep=0pt] {\scriptsize $u$}; 
  \node at (.5cm,0) (second) [draw,circle,fill=white,text width=7pt,align=center,
        inner sep=0pt] {\scriptsize $w$}; 
  \draw[->] (first.north east) to [out=45, in=135] node [above=-2bp] {\scriptsize
    $\stmtnoovfl{\pi}$} (second.north
  west);
  \draw[->] (first.south east) to [out=-45, in=-135] node [below=-2bp]
       {\scriptsize $\stmtovfl{\pi}$} (second.south west);
\end{tikzpicture}}
}

\newcommand{\safe}{\ding{52}}
\newcommand{\unsafe}{\ding{56}}
\newcommand{\incomplete}{?}
\newcommand{\crash}{\ensuremath{\lightning}}
\newcommand{\timeout}{T/O}

%% file: abstract.tex
\begin{abstract}
Loop under-approximation is a technique that enriches
C programs with additional branches that represent
the effect of a (limited) range of loop iterations. While
this technique can speed up the detection of bugs significantly,
it introduces redundant execution traces which may complicate the
verification of the program. This holds particularly true for
verification tools based on Bounded Model Checking, which 
incorporate simplistic heuristics to determine whether all
feasible iterations of a loop have been considered.

We present a technique that uses \emph{trace automata}
to eliminate redundant executions after performing loop
acceleration. The method reduces the diameter of the program
under analysis, which is in certain cases sufficient to
allow a safety proof using Bounded Model Checking. 
Our transformation is precise---it does not introduce
false positives, nor does it mask any errors. We have implemented the
analysis as a source-to-source transformation, and present experimental
results showing the applicability of the technique.
\end{abstract}

%% file: introduction.tex
\section{Introduction}
\label{sec:intro}

Software verification can be loosely divided into two themes: finding bugs
and proving correctness.  These two goals are often at odds with one
another, and it is rare that a tool excels at both tasks.  This tension is
well illustrated by the results of the 2014 Software Verification
Competition (\SVCOMP)~\cite{svcomp14}, in which several of the
best-performing tools were based on Bounded Model Checking
(BMC)~\cite{BiereCCZ99}.  The BMC-based tools were able to quickly find
bugs in the unsafe programs, but were unable to soundly prove safety for the
remaining programs.  Conversely, many of the sound tools had difficulty in
detecting bugs in the unsafe programs.

The reasons for this disparity are rooted in the very nature of contemporary
verification tools.  Tools aiming at proof typically rely on
over-approximating abstractions and refinement techniques to derive the loop
invariants required (e.g.,~\cite{hrmgs02,McMillan06}).  For certain classes
of programs, invariants can be found efficiently using
templates~\cite{BeyerHMR07} or theorem provers~\cite{KovacsV09FASE}.  For
unsafe programs, however, any attempt to construct a safety invariant must
necessarily fail, triggering numerous futile refinement iterations before a
valid counterexample is detected.  Verifiers based on the BMC paradigm (such
as \CBMC~\cite{ckl2004}), on the other hand, are able to efficiently detect
shallow bugs, but are unable to prove safety in most cases.


The key principle of this paper is that BMC is able to prove safety once the
unwinding bound exceeds the reachability diameter of the
model~\cite{BiereCCZ99,KroeningS03}.  The diameter of non-trivial programs is
however in most cases unmanageably large.  Furthermore, even when the
diameter is small, it is often computationally expensive to determine, as
the problem of computing the exact diameter is equivalent to a 2-QBF
instance.

The contribution of this paper is a technique that reduces the diameter of a
program in a way that the new, smaller diameter can be computed by means of
a simple satisfiability check. The technique has two steps:
\begin{enumerate}

\item We first identify potentially deep program paths that can be replaced
by a concise single-step summary called an
\emph{accelerator}~\cite{boigelot99,fl2002,BozgaIK10}.

\item We then remove those paths subsumed by the accelerators from the
program using \emph{trace automata}~\cite{HeizmannHP09}.

\end{enumerate}
The resulting program preserves the reachable states of the original
program, but is often very shallow, and consequently, we can obtain a sound
verification result using BMC.

Our paper is organised as follows: We present a number of
motivating examples and an outline of our approach in
Section~\ref{sec:motivation}. Section~\ref{sec:preliminaries} presents our
notation, recapitulates the concept of a
reachability diameter, and introduces 
a generalised notion of the under-approximating accelerators
presented in~\cite{KroeningLW13}. 
Section~\ref{sec:reduction} describes the construction
of accelerated programs and discusses the resulting reduction
of the reachability diameter of the program. 
In Section~\ref{sec:safety}, we introduce restricting
languages and trace automata as a means to eliminate
redundant transitions from accelerated programs.
The experimental evaluation based on a selection of
\SVCOMP benchmarks is presented in Section~\ref{sec:experiments}.
Finally, Section~\ref{sec:related} briefly surveys related work.

%% file: motivation.tex
\section{Motivation}
\label{sec:motivation}

In this section we will discuss the differences between proving safety and
finding bugs, with reference to some \SVCOMP benchmarks, and informally
demonstrate why our method is effective for both kinds of analyses.

The program in Figure~\ref{fig:safe}, taken from the \textsc{Loops}
category of \SVCOMP, proved challenging for many of the
participating tools, with only 6 out of the 12 entrants solving it correctly. 
A proof of safety for this program using an abstract interpreter requires a
relational domain to represent the invariant 
${\tt x}+{\tt y}={\tt N}$, which is often expensive.

The program in Figure~\ref{fig:unsafe} resembles the one
in Figure~\ref{fig:safe}, except for
the negated assertion at the end.  This example is very easy for Bounded
Model Checkers, which are able to discover a bug in a single unwinding by
assigning ${\tt N}=1$. A~slight modification, however,
illustrated in Figure~\ref{fig:deep}, increases the number of loop
iterations required to trigger the bug to $10^6$, exceeding the
capability of even the best BMC-based verification tools.

\begin{figure}\centering
  \parbox{.3\textwidth}{
    \fbox{
    \begin{minipage}{\textwidth}
      \begin{tabbing}
        \qquad\=\qquad\=\qquad\=\qquad\=\kill
        ${\tt unsigned~N:=*;}$\\
        ${\tt unsigned~x:=N,  y:= 0;}$\\
        ${\tt while~(x>0)~\{}$\\
        \>${\tt x:=x-1;}$\\ 
        \>${\tt y:=y+1;}$\\
        ${\tt\}}$\\
        ${\tt assert~(y=N)}$;
      \end{tabbing}
    \end{minipage}}
    \caption{Safe program\label{fig:safe}}
  }\quad
  \parbox{.3\textwidth}{\centering
    \fbox{
    \begin{minipage}{\textwidth}
      \begin{tabbing}
        \qquad\=\qquad\=\qquad\=\qquad\=\kill
        ${\tt unsigned~N=*;}$\\
        ${\tt unsigned~x:=N,  y:=0;}$\\
        ${\tt while~(x>0)~\{}$\\
        \>${\tt x:=x-1;}$\\ 
        \>${\tt y:=y+1;}$\\
        ${\tt\}}$\\
        ${\tt assert~(y\neq N)}$;
      \end{tabbing}
    \end{minipage}}    
    \caption{Unsafe program\label{fig:unsafe}}
  }\quad
  \parbox{.3\textwidth}{
    \fbox{
    \begin{minipage}{\textwidth}
      \begin{tabbing}
        \qquad\=\qquad\=\qquad\=\qquad\=\kill
        ${\tt unsigned~N:=10^6;}$\\
        ${\tt unsigned~x:=N, y:= 0;}$\\
        ${\tt while~(x>0)~\{}$\\
        \>${\tt x:=x-1;}$\\ 
        \>${\tt y:=y+1;}$\\
        ${\tt\}}$\\
        ${\tt assert~(y\neq N)}$;
      \end{tabbing}
    \end{minipage}}
    \caption{``Deep'' bug\label{fig:deep}}
  }
\end{figure}

The relative simplicity of the program statements in
Figures~\ref{fig:safe} to \ref{fig:deep} makes
them amenable to \emph{acceleration}~\cite{boigelot99,fl2002,BozgaIK10},
a technique used to compute the effect of the repeated iteration of 
statements over integer linear arithmetic. Specifically,
the effect of $i$ loop iterations is that {\tt x} is decreased
and {\tt y} is increased by $i$. 
Acceleration, however, is typically restricted to programs over
fragments of linear arithmetic for which the transitive
closure is effectively computable, thus restricting its applicability
to programs whose semantics can be soundly modelled using unbounded integers.
In reality, however,
the scalar variables in Figures~\ref{fig:safe} 
to \ref{fig:deep} take their values from the bounded
subset $\{0,\ldots,(2^{32}-1)\}$ of the positive integers $\mathds{N}_0$.
Traditional acceleration techniques do not account for integer
overflows. To address this problem, we previously introduced
\emph{under-approximate acceleration}, bounding the acceleration
to the interval in which the statements behave 
uniformly~\cite{KroeningLW13}. 

The code snippet in Figure~\ref{fig:acc_body} represents an 
under-approximating accelerator for the loop bodies in Figures~\ref{fig:safe},
\ref{fig:unsafe}, and \ref{fig:deep}.
\begin{figure}[t]\centering
\parbox{.55\textwidth}{\centering
  \begin{minipage}[t]{.52\textwidth}
    \begin{displaymath}
      \begin{aligned}
        &\left.\begin{array}{p{3.5cm}}
          ${\tt unsigned\;}i{\tt\,:=*;}$\\
          ${\tt assume\;(}i{\tt\,>0)}$\\
        \end{array}\right\}\text{iteration counter}\\
        &\left.\begin{array}{p{3.5cm}}
          ${\tt assume(x>0);}$\\
        \end{array}\right\}\text{feasibility check}\\
        &\left.\begin{array}{p{3.5cm}}
          ${\tt x:=x-}i{\tt;}$\\
          ${\tt y:=y+}i{\tt;}$\\
        \end{array}\right\}\text{acceleration}\\
        &\left.\begin{array}{p{3.5cm}}
          ${\tt assume(\neg{\sf underflow}\;(x));}$\\
        \end{array}\right\}\text{iteration bound}
      \end{aligned}
    \end{displaymath}
  \end{minipage}
  \caption{Accelerated loop body\label{fig:acc_body}}
}\;
\parbox{.42\textwidth}{\centering
  \fbox{\begin{minipage}[t]{\textwidth}
      \begin{tabbing}
        \qquad\=\qquad\=\qquad\=\qquad\=\kill
        ${\tt unsigned~N:=10^6, x:=N, y:= 0;}$\\
        ${\tt while~(x>0)~\{}$\\
        \>${\tt if~(*)~\{}$\\
        \>\>$i{\tt\,:=*;\;assume\;(}i{\tt\,>0);}$\\
        \>\>${\tt x:=x-}i{\tt;\; y=y+}i{\tt;}$\\
        \>\>${\tt assume\;(x\geq 0);}$\\
        \>${\tt\}\;else\;\{}$\\
        \>\>${\tt x:=x-1;\;y:=y+1;}$\\
        \>${\tt\}}$\\
        ${\tt\}}$\\
        ${\tt assert~(y\neq N)}$;
  \end{tabbing}\end{minipage}}
  \caption{Accelerated unsafe program\label{fig:accelerated1}}
}
\end{figure}
We introduce an auxiliary variable $i$ representing
a non-deterministic number of loop iterations. The
subsequent assumption guarantees that the accelerated
code reflects at least one iteration (and is optional
in this example). The assumption that follows warrants
the feasibility of the 
accelerated trace (in general, this condition
may contain quantifiers~\cite{KroeningLW13}). The
effect of $i$ iterations is encoded using the two
assignment statements, which constitute the closed
forms of the recurrence relations corresponding to
the original assignments. The final assumption
guarantees that $i$ lies in the range in which
the right-hand sides of the assignments behave linearly.

In general, under-approximating accelerators do not reflect
all feasible iterations of the loop body. Accordingly,
we cannot simply replace the original loop body. 
Instead, we add back the accelerator as an additional
path through the loop, as illustrated in 
Figure~\ref{fig:accelerated1}.

The transformation preserves safety properties---that is to say, 
an accelerated program has a reachable, failing assertion iff the
original program does. We can see that the failing assertion in
Figure~\ref{fig:accelerated1} is reachable after a single iteration of the
loop, by simply choosing $i={\tt N}$.
Since the accelerated program contains a feasible trace leading to a failed
assertion, we can conclude that the original program does as well, despite
having only considered a single trace of length~1.

\begin{figure}[t]\centering
  \parbox{.4\textwidth}{\centering
    \fbox{\begin{minipage}[b]{.4\textwidth}
        \begin{tabbing}
          \qquad\=\qquad\=\qquad\=\qquad\=\kill
          ${\tt unsigned~N:=10^6, x:=N, y:= 0;}$\\
          ${\tt if~(x>0)~\{}$\\
          \>${\tt x:=x-1;\; y:=y+1;}$\\
          \>${\tt if~(x>0)~\{}$\\
          \>\>${\tt x:=x-1;\; y:=y+1;}$\\
          \>\>${\tt if~(x>0)~\{}$\\
          \>\>\>${\tt x:=x-1;}$\\
          \>\>\>${\tt y:=y+1;}$\\
          \>\>\>${\tt assert\;(x\leq 0);}$\\
          \>\>${\tt\}}$\\
          \>${\tt\}}$\\
          ${\tt\}}$\\
          ${\tt assert\;(y=N)}$;
    \end{tabbing}\end{minipage}}
  \caption{Unwinding ($k=3$) of safe program with 
    ${\tt N=10^6}$\label{fig:unwinding}}}\qquad %
  \parbox{.5\textwidth}{\centering
    \fbox{\begin{minipage}[b]{.5\textwidth}
        \begin{tabbing}
          \qquad\=\qquad\=\qquad\=\qquad\=\kill
          \>${\tt unsigned~N:=*, x:=N, y:= 0;}$\\
          \>${\tt bool~g:=*;}$\\
          1:\>${\tt while~(x>0)~\{}$\\
          \>\>${\tt if~(*)~\{}$\\
          \>\>\>${\tt assume\;(\neg g);}$\\
          2:\>\>\>$i{\tt\,:=*;\;x:=x-}i{\tt;\; y=y+}i{\tt;}$\\
          \>\>\>${\tt assume\;(x\geq 0);}$\\
          3:\>\>\>${\tt g:=\true;}$\\
          \>\>${\tt\}\;else\;\{}$\\
          \>\>\>${\tt x:=x-1;\;y:=y+1;}$\\
          \>\>\>${\tt assume\;({\sf underflow}\;(x));}$\\
          \>\>\>${\tt g:=\false;}$\\
          \>\>${\tt\}}$\\
          \>${\tt\}}$\\
          4:\>${\tt assert~(y=N)}$;
    \end{tabbing}\end{minipage}}
    \caption{Accelerated and instrumented safe 
      program\label{fig:restricted1}}}
\end{figure}

While the primary application of BMC is bug detection, contemporary
Boun\-ded Model Checkers such as \CBMC are able to prove safety
in some cases. \CBMC unwinds loops up to a predetermined 
bound $k$ (see Figure~\ref{fig:unwinding}).
 \emph{Unwinding assertions} are one possible mechanism
to determine whether further unwinding is
required~\cite{ckl2004,dkw2008}. The assertion ${\tt(x\leq 0)}$
in Figure~\ref{fig:unwinding} fails if there are feasible
program executions traversing the loop more than three times.
It is obvious that this assertion will fail for any $k<10^6$.

Unfortunately, acceleration is ineffective in this setting.
Since the accelerator in Figure~\ref{fig:accelerated1}
admits $i=1$, we have to consider $10^{6}$ unwindings
before we can establish the safety of the program
in Figure~\ref{fig:safe} with ${\tt N=10^6}$. For a
non-deterministically assigned ${\tt N}$, this number increases to $2^{32}$.

This outcome is disappointing, since the repeated iteration
of the accelerated loop body is redundant. Furthermore,
there is no point in taking the unaccelerated path through 
the loop (unless there is an impending overflow---which can
be ruled out in the given program), since 
the accelerator \emph{subsumes} this execution (with $i=1$).
Thus, if we eliminate all executions that meet either of the criteria above,
we do not alter the semantics of the program but may reduce the difficulty
of our problem considerably.

Figure~\ref{fig:restricted1} shows an accelerated 
version of the safe program of Figure~\ref{fig:safe}, but 
instrumented to remove redundant traces. This is achieved
by introducing an auxiliary variable ${\tt g}$ which
determines whether the accelerator was traversed in the
previous iteration of the loop. This flag is reset
in the non-accelerated branch, which, however, in
our example is never feasible.
It is worth noting that every feasible trace
through Listing~\ref{fig:safe} has a corresponding feasible
trace through Listing~\ref{fig:restricted1}, and vice versa.\\

\noindent%
\parbox{.63\textwidth}{
  The figure to the right shows an execution 
  of the program in Figure~\ref{fig:restricted1}:
  This trace is both feasible and safe---the 
  assertion on line~4 is not violated.
  It is not too difficult to see that \emph{every} feasible trace through
  the program in Figure~\ref{fig:restricted1} has the same length, 
  which means that we can soundly reason about its
  safety considering traces with a single iteration of the loop, 
  which is a tractable (and indeed, easy) problem.}\quad
\parbox{.3\textwidth}{
  \fbox{
    \begin{tabular}{lllllc}
      Loc. & ${\tt N}$ & ${\tt x}$ & ${\tt y}$ & $i$ & ${\tt g}$ \\
      \hline
      1      & $10^4$ & $10^4$ & 0 & 0 & \false \\
      2      & $10^4$ & $10^4$ & 0 & 0 & \false \\
      3      & $10^4$ & 0       & $10^4$ & $10^4$ & \false \\
      1      & $10^4$ & 0       & $10^4$ & $10^4$ & \true \\
      4      & $10^4$ & 0       & $10^4$ & $10^4$ & \true
  \end{tabular}
  }
}\\

Since the accelerated and instrumented
program in Figure~\ref{fig:restricted1}
is safe, we can conclude that the original program
in Figure~\ref{fig:safe} is safe as well.

We emphasise that our approach neither introduces an
over-approximation, nor requires the explicit
computation of a fixed point. In addition, it is
not restricted to linear integer arithmetic and bit-vectors:
our prior work can generate some non-linear accelerators
and also allows for the acceleration of a limited
class of programs with arrays~\cite{KroeningLW13}.

%% file: preliminaries.tex
\section{Notation and Basic Concepts}
\label{sec:preliminaries}

\begin{table}\centering
  \caption{Program Statements and Traces}
  \subfigure[Syntax and Semantics\label{tbl:stmts}]{
    \begin{minipage}{.4\textwidth}
    \begin{gather*}
      {\tt stmt} ::= {\tt x}:=e \;\vert\; [B] \;\vert\; \skipstmt \\
      ({\tt x}\in\Vars,\; e\in\Exprs,\; B\in\BExprs)\\
      \begin{array}{lcl}
        \wlpcond{{\tt x}:=e}{P}&\defn& P[e/ {\tt x}] \\
        \wlpcond{{\tt x}:=*}{P}&\defn& \forall {\tt x}\,.\, P \\
        \wlpcond{[B]}{P}&\defn& B \Rightarrow P\\
        \wlpcond{\skipstmt}{P} &\defn& P \\
      \end{array}
    \end{gather*}
    \end{minipage}
  }\quad
  \subfigure[Transition Relations for Traces\label{tbl:trans}]{
    \begin{minipage}{.5\textwidth}
    \begin{gather*}
      \begin{array}{rcl}
      \sem{\stmt}&\defn&\neg\wlpcond{\stmt}{\bigvee_{{\tt x}\in\Vars} {\tt x}\neq{\tt
          x}'}\\
      \id&\defn &\sem{\skipstmt}\\
      \sem{\stmt_1\cdot\stmt_2}&\defn&\sem{\stmt_1}\circ\sem{\stmt_2}\\
      \sem{\stmt^n}&\defn& \sem{\stmt}^n,\\
      \end{array}\\
      \text{~where~} 
      \begin{array}{rcl}
        \stmt^0&\defn&\emptypath,\\
        \stmt^n&\defn&\stmt\cdot(\stmt^{(n-1)}) \\
        \sem{\stmt}^0&\defn&\id,\\
        \sem{\stmt}^n&\defn&
        \sem{\stmt}\circ(\sem{\stmt}^{(n-1)})\\
      \end{array}
    \end{gather*}
    \end{minipage}
  }
\end{table}

Let \Stmts be the (infinite) set of statements of a simple programming
language as defined in Table~\ref{tbl:stmts}, where
\Exprs and \BExprs denote expressions and predicates
over the program variables \Vars, respectively. Assumptions
are abbreviated by $[B]$, and assertions are modeled using
assumptions and error locations. For brevity, we omit
array accesses. We assume that different occurrences of statements 
are distinguishable (using the program locations).
The semantics is provided by the weakest liberal precondition 
\emph{wlp} as defined in~\cite{Nelson89}. Programs are represented
using control flow automata.

\begin{definition}[CFA] A \emph{control flow automaton} $P$
  is a directed graph $\langle V, E, v_0\rangle$,
  where $V$ is a finite set of vertices,
  $\PStmts\subseteq\Stmts$ is a finite set of statements,
  $E\subseteq\left(V\times\PStmts\times V\right)$ is 
  a set of edges, and $v_0\in V$ is the initial vertex.
  We write $\trans{v}{u}{\stmt}$ if
  $\langle u,\stmt,v\rangle\in E$.
  \label{def:cfa}
\end{definition}

A program state $\state$ is a total function assigning a value
to each program variable in $\Vars$. $\States$ denotes the set 
of program states. A transition relation $T\subseteq\States\times\States$
associates states with their successor states. Given $\Vars$,
let $\Vars'$ be a corresponding set of primed variables encoding
successor states.
The symbolic transition relation for a statement or trace is a predicate over 
$\Vars\cup\Vars'$ and can be derived using \emph{wlp}
as indicated in Table~\ref{tbl:trans} (cf.~\cite{ewd821}).
We write $\langle\state,\state'\rangle\in\sem{\stmt}$ if
$\sem{\stmt}$ evaluates to true under $\state$ and $\state'$
(i.e., $\state,\state'\models\sem{\stmt}$).
A trace $\pi$ is \emph{feasible} if there exist states
$\sigma,\sigma'$ such that $\langle\state,\state'\rangle\in\sem{\pi}$.

Given a CFA $P\defn\langle V, E, v_0\rangle$,
a trace $\pi\defn\stmt_{i}\cdot\stmt_{i+1}\cdots\stmt_{n}$ 
(where $\trans{v_{j-1}}{v_j}{\stmt_j}$ for $i<j\leq n$) 
of length $\vert\pi\vert=n-i+1$ is
\emph{looping} (with head $v_i$) 
iff $v_i=v_n$, and \emph{accepted} by the CFA 
iff $v_i=v_0$. 
We use $\lang_P$ to denote the set of all traces
that are accepted by the CFA $P$.
Abusing our notation, we write 
$\trans{v_{i}}{v_j}{\pi}$ to denote path 
starting at $v_i$ and ending at $v_j$ and corresponding
to the trace $\pi$.
  
A state $\state$ is \emph{reachable} from an initial state $\state_0$
iff there exists a trace $\pi$ accepted by the CFA such that
$\langle\sigma_0,\sigma\rangle\in\sem{\pi}$. The reachability
diameter~\cite{BiereCCZ99,KroeningS03} of a transition relation is the 
smallest number of steps required to reach all reachable states:
\begin{definition}[Reachability Diameter]
  \label{def:rd}
  Given a CFA with initial state $\state_0$, the 
  \emph{reachability diameter}
  is the smallest $n$ such that for every 
  state $\state$ reachable from $\state_0$ 
  there exists a feasible trace $\pi$ of length
  at most $n$ accepted by the CFA with
  $\langle\state_0,\state\rangle\in\sem{\pi}$.
\end{definition}

To show that a CFA does not violate a given
safety (or reachability) property, it is
sufficient to explore all feasible traces whose length
does not exceed the reachability diameter.
In the presence of looping traces, however, 
the reachability diameter of a program
can be infinitely large. 

Acceleration~\cite{boigelot99,fl2002,BozgaIK10} is a technique
to compute the reflexive transitive closure 
$\sem{\pi}^*\defn\bigcup_{i=0}^{\infty}\sem{\pi}^i$
for a looping trace $\pi$. Equivalently, 
$\sem{\pi}^*$ can be expressed as $\exists
i\in\mathds{N}_0\,.\,\sem{\pi}^i$. The aim
of acceleration is to express $\sem{\pi}^*$ in a
decidable fragment of logic. In general,
this is not possible, even if $\sem{\pi}$ is
defined in a decidable fragment of integer
arithmetic such as Presburger arithmetic.
For octagonal relations $\sem{\pi}$, however,
the transitive closure is $\sem{\pi}^*$ is
Presburger-definable and effectively
computable~\cite{boigelot99,fl2002}. 

\begin{definition}[Accelerated Transitions]
  \label{def:acc}
  Given a looping trace~$\pi\in\lang_P$, we say that a trace
  $\hat{\pi}\in\Stmts^*$
  is an \emph{accelerator} for $\pi$ if
  $\sem{\acc{\pi}}\equiv \sem{\pi}^*$.

  An accelerator $\uacc{\pi}\in\Stmts^*$ is \emph{under-approximating}
  if the number of iterations is bounded from above by 
  a function $\beta:\States\rightarrow\mathds{N}_0$ 
  of the starting state $\state$:
  \begin{displaymath}
    \langle\state,\state'\rangle\in\sem{\uacc{\pi}}\quad\text{iff}\quad
    \exists i\in\mathds{N}_0\,.\,
    i\leq\bound{\state}\wedge\langle\state,\state'\rangle\in\sem{\pi}^i
  \end{displaymath}
  We require that the function $\beta$ has the following property:
  \begin{equation}
    \label{eq:bound}
    \left(
      i\leq\bound{\state}\wedge\langle\state,\state'\rangle\in\sem{\pi}^i
    \right)\Rightarrow
    \left(\bound{\state'}\leq\bound{\state}-i\right)
  \end{equation}

  We say that $\uacc{\pi}$ is strictly under-approximating
  if $\sem{\uacc{\pi}}\subset\sem{\acc{\pi}}$.
\end{definition}

We introduced under-approximating accelerators 
for linear integer arithmetic and the theories of bit-vectors and arrays 
in~\cite{KroeningLW13} in order
to accelerate the detection of counterexamples. 
Under-approximations are caused by transition relations that
can only be accelerated within certain intervals, e.g., the range
in which no overflow occurs in the case of bit-vectors, or in
which no conflicting assignments to array elements are made. 
The bound function $\beta$ restricts this interval accordingly.

\begin{example} An under-approximating accelerator 
for the statement ${\tt x:=x+1}$, where
${\tt x}$ is a 32-bit-wide unsigned integer, can be given as
\begin{displaymath}
\uacc{\pi}\defn i:=*{\tt ;\,[x}+i< 2^{32}{\tt ]; x}:={\tt x}+i
\end{displaymath}
with transition relation $\exists i\,.\,\left({\tt x}+i<2^{32}\right)\wedge
\left({\tt x}'={\tt x}+i\right)$. Note that $\beta$ is implicit
here and that the alphabet of $\uacc{\pi}$ is not restricted to
$\PStmts$.
\end{example}

%% file: reduce_diameter.tex
\section{Diameter Reduction via Acceleration}
\label{sec:reduction}

In this section, we introduce a reachability-preserving
program transformation that reduces the reachability
diameter of a CFA. While a similar transformation is
used in~\cite{KroeningLW13} to detect counterexamples
with loops, our goal here is to reduce the diameter
in order to enable safety proofs (see Section~\ref{sec:safety}).

\begin{definition}[Accelerated CFA]
  \label{def:acc_cfa}
  Let $P\defn\langle V, E, v_0\rangle$ be a CFA over
  the alphabet $\PStmts$, and let $\pi_1,\ldots,\pi_k$ be 
  traces in $P$ looping with heads $v_1,\ldots,v_k\in V$, respectively.
  Let $\acc{\pi}_1,\ldots\acc{\pi}_k$ be the
  (potentially under-approximating) accelerators for 
  $\pi_1,\ldots,\pi_k$. Then the
  \emph{accelerated CFA}
  $\acc{P}\defn\langle \acc{V}, \acc{E}, v_0\rangle$
  for $P$ is the CFA $P$ augmented with non-branching paths
  $\trans{v_i}{v_i}{\acc{\pi}_i}$
  ($1\leq i \leq k$).
\end{definition}


A trace is \emph{accelerated} if it traverses a path in
$\acc{P}$ that corresponds to an accelerator.
A trace $\pi_1$ \emph{subsumes} a
trace $\pi_2$, denoted by $\pi_2\preceq\pi_1$,
if $\sem{\pi_2}\subseteq\sem{\pi_1}$. Accordingly, 
$\pi\preceq\acc{\pi}$ and $\uacc{\pi}\preceq\acc{\pi}$
(by Definition~\ref{def:acc}). We extend the 
relation $\preceq$ to sets of traces: 
$\Pi_1\preceq\Pi_2$ if 
$\left(\bigcup_{\pi\in\Pi_1}\sem{\pi}\right)\preceq
\left(\bigcup_{\pi\in\Pi_2}\sem{\pi}\right)$. A trace $\pi$
is \emph{redundant} if $\{\pi\}$ is subsumed by a set
$\Pi\setminus\{\pi\}$ of other traces in the CFA.

\begin{lemma}
  \label{lem:subsume_iteration}
  Let $\uacc{\pi}$ be an under-approximating accelerator
  for the looping trace $\pi$. Then 
  $\uacc{\pi}\cdot\uacc{\pi}\preceq\uacc{\pi}$ holds.
\end{lemma}

A proof is provided in Appendix~\ref{sec:proofs}.
The following theorem states that the transformation
in Definition~\ref{def:acc_cfa} preserves the
reachability of states and never increases the reachability
diameter.

\begin{theorem}
\label{thm:red_rd}
Let $P$ be a CFA and $\acc{P}$ a corresponding
accelerated CFA as in Definition~\ref{def:acc_cfa}. Then 
the following claims hold:
\begin{enumerate}
\item Every trace in $P$ is subsumed by at least one trace in $\acc{P}$.
\item Let $\pi_1$ be an accelerated trace accepted by $\acc{P}$,
  and let $\langle\state_0,\state\rangle\in\sem{\pi_1}$. Then
  there exists a trace $\pi_2$ accepted by $P$ such that
  $\langle\state_0,\state\rangle\in\sem{\pi_2}$.
\end{enumerate}
\end{theorem}

\begin{proof}
  Part 1 of the theorem holds because $P$ is a sub-graph
  of $\acc{P}$. For the second
  part, assume that $\acc{\pi}_1,\ldots\acc{\pi}_k$ are
  the accelerators occurring in $\pi_1$. Then there
  are $i_1,\ldots,i_k\in\mathds{N}$ such that 
  $\pi_2\defn\pi_1[\pi_1^{i_1}/\acc{\pi}_1]\cdots[\pi_k^{i_k}/\acc{\pi}_k]$
  and $\langle\state_0,\state\rangle\in\sem{\pi_2}$.
\end{proof}

The diameter of a CFA
is determined by the longest of the shortest traces
from the initial state $\state_0$ to all 
reachable states~\cite{KroeningS03}. Accordingly,
the transformation in Definition~\ref{def:acc_cfa}
results in a reduction of the diameter if it
introduces a shorter accelerated trace that results
in the redundancy of this longest shortest trace.
In particular, acceleration may reduce an infinite
diameter to a finite one.

%% file: safety_checking.tex
\section{Checking Safety with Trace Automata}
\label{sec:safety}

Bounded Model Checking owes its industrial success largely to its
effectiveness as a bug-finding technique. Nonetheless,
BMC can also be used to prove safety properties if the 
unwinding bound exceeds the reachability
diameter. In practice, however, the diameter 
can rarely be determined statically. Instead,
\emph{unwinding assertions} are used to detect looping traces that
become infeasible if expanded further~\cite{ckl2004}. Specifically, 
an unwinding assertion is a condition that fails for an
unwinding bound $k$ and a
trace $\pi_1\cdot \pi_2^k$ if $\pi_1\cdot \pi_2^{k+1}$ is feasible,
indicating that further iterations may be required to 
exhaustively explore the state space.

In the presence of accelerators, 
however, unwinding assertions are 
inefficient. Since $\acc{\pi}\cdot\acc{\pi}\preceq\acc{\pi}$
(Lemma~\ref{lem:subsume_iteration}),
repeated iterations of accelerators are redundant. The unwinding
assertion for $\pi_1\cdot\acc{\pi}_2$, however, fails if
$\pi_1\cdot\acc{\pi}_2\cdot\acc{\pi}_2$ is feasible. Accordingly,
the approximate diameter as determined by means of unwinding
assertions for an accelerated program $\acc{P}$
is the \emph{same} as for the corresponding non-accelerated
program $P$. 

In the following, we present a technique that remedies
the deficiency of unwinding assertions in the presence
of accelerators by \emph{restricting} the language
accepted by a CFA.

\begin{definition}[Restriction Language]
  \label{def:restr_lang}
  Let $\acc{P}$ an accelerated CFA for $P$ over the vocabulary
  $\AccPStmts$. For each accelerator 
  $\acc{\pi}\in\AccPStmts^+$, 
  let $\pi\in\PStmts^+$ be the corresponding looping trace.
  The \emph{restriction language} $\lang_R$ for $\acc{P}$ 
  comprises all traces with a sub-trace characterised
  by the regular expression
  $(\pi\,\vert\,(\acc{\pi}\cdot\acc{\pi}))$ for all
  accelerators $\acc{\pi}$ in $\acc{P}$ with $\pi\preceq\acc{\pi}$.
\end{definition}

The following lemma enables us to eliminate traces
of an accelerated CFA $\acc{P}$ that are in 
the restriction language $\lang_R$.

\begin{lemma}
  \label{lem:restr_trace}
  Let $\acc{P}$ be an accelerated CFA, and
  $\lang_R$ be the corresponding restriction language.
  Let $\pi_1$ be a trace accepted by $\acc{P}$ such
  that $\pi_1\in\lang_R$. 
  Then there exists a trace $\pi_2$ which is accepted by $\acc{P}$
  such that $\pi_1\preceq\pi_2$ and $\pi_1$ is not a sub-trace
  of~$\pi_2$.
\end{lemma}

A proof by case split is provided in Appendix~\ref{sec:proofs}.
Using Lemma~\ref{lem:restr_trace} and induction over the 
number of traces and accelerators, it is admissible to eliminate all traces 
accepted by $\acc{P}$ and contained in $\lang_R$ without affecting
the reachability of states:

\begin{theorem}
  \label{thm:restr_cfa}
  Let $\lang_{\acc{P}}$ be the language comprising all traces
  accepted by an accelerated CFA $\acc{P}$ and $\lang_R$ be
  the corresponding restriction language. Then
  every trace $\pi\in\lang_{\acc{P}}$ is subsumed by
  the traces in $\lang_{\acc{P}}\setminus\lang_R$.
\end{theorem}

Notably, Definition~\ref{def:restr_lang} explicitly 
excludes accelerators $\acc{\pi}$ that do not 
satisfy $\pi\preceq\acc{\pi}$, a requirement
that is therefore implicitly present in
Lemma~\ref{lem:restr_trace} as well as 
Theorem~\ref{thm:restr_cfa}. The rationale
behind this restriction is that 
strictly under-approximating accelerators $\uacc{\pi}$
do not necessarily have this property. 
However, even if $\uacc{\pi}$ does not subsume 
$\pi$ in general, we can characterize
the set of starting states in which it does:
\begin{equation}
  \left\{\state\,\vert\,
    \langle\state,\state'\rangle\in\sem{\pi}
  \Rightarrow
  \langle\state,\state'\rangle\in\sem{\uacc{\pi}}\right\}
\end{equation}

In order to determine whether a looping path $\pi$
is redundant, we presume for each
accelerated looping trace $\pi$ the existence of a 
predicate $\predovfl{\pi}\in\Exprs$ and an assumption
statement $\stmtnoovfl{\pi}\defn[\predovfl{\pi}]$ such that
\begin{equation}
  \sem{\,\stmtnoovfl{\pi}\,}\defn
  \left\{\langle\state,\state\rangle\vert\,
    \langle\state,\state'\rangle\in\sem{\pi}
  \Rightarrow
  \langle\state,\state'\rangle\in\sem{\uacc{\pi}}\right\}
\end{equation}

Analogously, we can define the dual statement
$\stmtovfl{\pi}\defn[\neg\predovfl{\pi}]$. 
Though both $\sem{\stmtnoovfl{\pi}}$ and $\sem{\stmtovfl{\pi}}$ 
are non-total transition relations, their combination 
$\sem{\stmtnoovfl{\pi}}\cup\sem{\stmtovfl{\pi}}$ is total. 
Moreover, it does not modify the state, i.e.,
$\sem{\stmtnoovfl{\pi}}\cup\sem{\stmtovfl{\pi}}
\equiv\sem{\skipstmt}$. 
It is therefore evident
that replacing the head $v$ of a looping trace $\pi$
with the sub-graph $\fork$ (and reconnecting the
incoming and outgoing edges of $v$ to $u$ and $w$,
respectively) preserves the reachability of states.
It does, however change the traces of the CFA. After
the modification, the looping traces 
$\stmtnoovfl{\pi}\cdot\pi$ and 
$\stmtovfl{\pi}\cdot\pi$ replace~$\pi$. 
By definition of $\stmtnoovfl{\pi}$,
we have $\stmtnoovfl{\pi}\cdot\pi\preceq\uacc{\pi}$. 
Consequently, if we accelerate
the newly introduced looping trace $\stmtnoovfl{\pi}\cdot\pi$, 
Definition~\ref{def:restr_lang} and therefore
Lemma~\ref{lem:restr_trace} as well as 
Theorem~\ref{thm:restr_cfa} apply.

The discriminating statement
$\stmtovfl{\pi}$ for the example path ${\tt x}:={\tt x}+1$
at the end of Section~\ref{sec:preliminaries}, for instance,
detects the presence of an overflow. For this specific
example, $\stmtovfl{\pi}$ is the assumption $[{\tt x}=2^{32}-1]$.
In practice, however, the bit-level-accurate encoding of \CBMC
provides a mechanism to detect an overflow \emph{after} 
it happened. Therefore, we introduce statements 
$\stmtovfl{\pi}\defn[\mathsf{overflow}({\tt x})]$ 
and $\stmtnoovfl{\pi}\defn[\neg\mathsf{overflow}({\tt x})]$ 
that determine the presence of an overflow
at the end of the looping trace. The modification and correctness
argument for this construction is analogous to the one above.

In order to recognize redundant traces,
we use a \emph{trace automaton} that accepts the restriction 
language $\lang_R$. 

\begin{definition}[Trace Automaton]
  \label{def:trace_automaton}
  A trace automaton $T_R$ for $\lang_R$ is a deterministic
  finite automaton (DFA) over the alphabet
  $\Stmts_{\acc{P}}$ that accepts $\lang_R$.
\end{definition}

Since $\lang_R$ is regular, so is its complement $\overline{\lang}_R$.
In the following, we describe an instrumentation of a CFA $\acc{P}$ which 
guarantees that every trace accepted by $T_R$ and $\acc{P}$ becomes infeasible.
To this end, we construct a DFA $T_R$ recognising $\lang_R$,
starting out with an $\epsilon$-NFA which we then determinise
using the subset construction~\cite{dragonbook}. While
this yields (for a CFA with $k$ statements) a DFA with $O(2^k)$ 
states in the worst case, in practice the DFAs generated are much
smaller.

We initialise the set the vertices of the instrumented CFA $\tilde{P}$ 
to the vertices of $\acc{P}$. We inline $T_R$ by
creating a fresh integer variable ${\tt g}$ in $\tilde{P}$
which encodes the state of $T_R$ and is initialised 
to $0$. For each edge $\trans{u}{v}{s} \in \acc{P}$, we consider all 
transitions $\trans{n}{m}{s} \in T_R$.  If there are no such transitions,
we copy the edge $\trans{u}{v}{s}$ into $\tilde{P}$. Otherwise,
we add edges as follows:
\begin{itemize}
 \item If $m$ is an accepting state, we do not add an edge to $\tilde{P}$.
 \item Otherwise, construct a new statement $l \defn [{\tt g} = n];
   {\tt g}:= m;\; s$ and
  add the path $\trans{u}{v}{l}$ to $\tilde{P}$, which simulates the
  transition $\trans{n}{m}{s}$.
\end{itemize}

Since we add at most one edge to $\tilde{P}$ for each transition in $T_R$,
this construction's time and space complexity are both
$\Theta(\| \acc{P} \| + \| T_R \|)$.
By construction, if a trace $\pi$ accepted by CFA $\tilde{P}$ 
projected to $\AccPStmts$ is contained in the restriction 
language $\lang_R$, then $\pi$ is infeasible.
Conceptually, our construction
suppresses traces accepted by $\lang_R$ and retains the 
remaining executions.

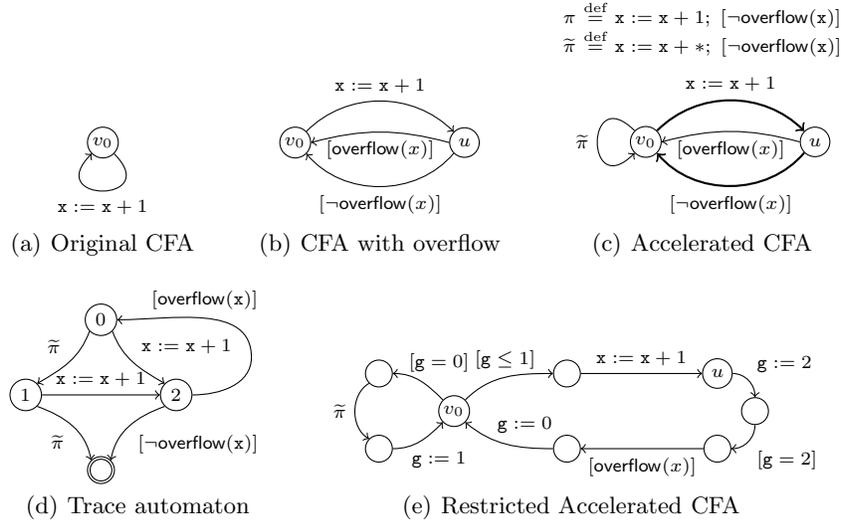
\begin{figure}[t]\centering
  \subfigure[Original CFA\label{fig:orig_cfa}]{
    \begin{minipage}[b]{.25\textwidth}\centering
      \begin{tikzpicture}
        \node at (0,0) (first) [draw,circle,fill=white,text width=8pt,align=center,inner sep=1pt] {\scriptsize $v_0$};
        \draw[->] (first.south east) to [out=-45,in=-135,looseness=8] 
        node[below] {\scriptsize ${\tt x}:={\tt x}+1$}
        (first.south west);
      \end{tikzpicture}
    \end{minipage}
  } 
  \subfigure[CFA with overflow\label{fig:discr_cfa}]{
    \begin{minipage}[b]{.3\textwidth}\centering
      \begin{tikzpicture}
        \node at (0,0) (first) [draw,circle,fill=white,text width=8pt,align=center,inner sep=1pt] {\scriptsize $v_0$};
        \node at (2.25cm,0cm) (second) [draw,circle,fill=white,text width=8pt,align=center,inner sep=1pt] {\scriptsize $u$};
        \draw[->] (first.north east) to [out=45,in=135] 
        node[above] {\scriptsize ${\tt x}:={\tt x}+1$}
        (second.north west);
        \draw[->] (second.west) to [out=165, in=15]
        node[below] {\scriptsize $[\mathsf{overflow}(x)]$}
        (first.east);
        \draw[->] (second.south west) to [out=230,in=-50]
        node[below] {\scriptsize $[\neg\mathsf{overflow}(x)]$}
        (first.south east);
      \end{tikzpicture}
    \end{minipage}
  } 
  \subfigure[Accelerated CFA\label{fig:acc_cfa}]{
    \begin{minipage}[b]{.35\textwidth}\centering
      \begin{tikzpicture}
        \node at (0,0) (first) [draw,circle,fill=white,text width=8pt,align=center,inner sep=1pt] {\scriptsize $v_0$};
        \node at (2.25cm,0cm) (second) [draw,circle,fill=white,text width=8pt,align=center,inner sep=1pt] {\scriptsize $u$};
        \draw[->,thick] (first.north east) to [out=45,in=135] 
        node[above] {\scriptsize ${\tt x}:={\tt x}+1$}
        (second.north west);
        \draw[->] (second.west) to [out=165, in=15]
        node[below] {\scriptsize $[\mathsf{overflow}(x)]$}
        (first.east);
        \draw[->,thick] (second.south west) to [out=230,in=-50]
        node[below] {\scriptsize $[\neg\mathsf{overflow}(x)]$}
        (first.south east);
        \draw[->] (first.north west) to [out=135,in=225,looseness=8] 
        node[left] {\scriptsize $\uacc{\pi}$}
        (first.south west);
        \node at (.75,1.5cm) {
          \scriptsize$\begin{array}{ccl}
            \pi & \defn & {\tt x}:={\tt x}+1;\;[\neg\sf{
                overflow}({\tt x})]\\
            \uacc{\pi} & \defn & {\tt x}:={\tt x}+*;\;[\neg\sf{
                overflow}({\tt x})]\\
          \end{array}$
        };
      \end{tikzpicture}
    \end{minipage}
  }\\[1ex]
  \subfigure[Trace automaton\label{fig:trace_aut}]{
    \begin{minipage}[b]{.3\textwidth}\centering
      \begin{tikzpicture}
        \node at (0,0) (first) [draw,circle,fill=white,text
          width=8pt,align=center,inner sep=1pt] {\scriptsize $0$};
        \node at (-1cm,-1cm) (second) [draw,circle,fill=white,text
          width=8pt,align=center,inner sep=1pt] {\scriptsize $1$};
        \node at (1cm,-1cm) (third) [draw,circle,fill=white,text
          width=8pt,align=center,inner sep=1pt] {\scriptsize $2$};
        \node at (0cm,-2cm) (final) [draw,circle,fill=white,text
          width=8pt,align=center,inner sep=1pt] {};
        \node at (0cm,-2cm) [draw,circle,fill=white,text
          width=6pt,align=center,inner sep=1pt] {};
        
        \draw[->] (second.east)--node[above] {\scriptsize ${\tt x}:={\tt x}+1$}
        (third.west);
        \draw[->] (first.south west) to [out=250,in=30]
        node[above left] {\scriptsize $\uacc{\pi}$}
        (second.north east);
        \draw[->] (first.south east) to [out=290,in=150]
        node[above right] {\scriptsize ${\tt x}:={\tt x}+1$}
        (third.north west);
        \draw[->] (second.south east) to [out=-20,in=100]
        node[below left] {\scriptsize $\uacc{\pi}$}
        (final.north west);
        \draw[->] (third.south west) to [out=200,in=80]
        node[below right] {\scriptsize $[\neg\sf{overflow}({\tt x})]$}
        (final.north east);
        \draw[->] (third.east) to [out=0,in=270] (2cm,-.5cm) to
             [out=90,in=0]  node[above] {\scriptsize$[\sf{overflow}({\tt x})]$}
             (first.east);
             
      \end{tikzpicture}
    \end{minipage}
  } 
  \subfigure[Restricted Accelerated CFA \label{fig:restrict_cfa}]{
    \begin{minipage}[b]{.6\textwidth}\centering
      \begin{tikzpicture}
        \node at (0,0) (first) [draw,circle,fill=white,text width=8pt,align=center,inner sep=1pt] {\scriptsize $v_0$};
        \node at (3.5cm,0.5cm) (second) [draw,circle,fill=white,text
          width=8pt,align=center,inner sep=1pt] {\scriptsize $u$};
        \node at (3.5cm,-0.5cm) (extra) [draw,circle,fill=white,text
          width=8pt,align=center,inner sep=1pt] {};
        \node at (4cm,0cm) (more) [draw,circle,fill=white,text
          width=8pt,align=center,inner sep=1pt] {};
        \node at (1.5cm,.5cm) (inter) [draw,circle,fill=white,text width=8pt,align=center,inner sep=1pt] {};
        \node at (1.5cm,-.5cm) (third) [draw,circle,fill=white,text width=8pt,align=center,inner sep=1pt] {};
        \draw[->] (inter.east) to [out=0,in=180] 
        node[above] {\scriptsize ${\tt x}:={\tt x}+1$}
        (second.west);
        \draw[->] (first.north east) to [out=45,in=180] 
        node[above] {\scriptsize $[{\tt g}\leq 1]$}
        (inter.west);
        \draw[->] (extra.west) to [out=180,in=0]
        node[below] {\scriptsize $[\mathsf{overflow}(x)]$}
        (third.east);
        \draw[->] (third.west) to [out=180,in=-40]
        node[shift={(.25cm,.25cm)}] {\scriptsize ${\tt g}:=0$}
        (first.south east);
        \node at (-1cm,.5cm) (a) [draw,circle,fill=white,text
          width=8pt,align=center,inner sep=1pt] {};
        \node at (-1cm,-.5cm) (d) [draw,circle,fill=white,text
          width=8pt,align=center,inner sep=1pt] {};
        \draw[->] (first.north west) to [out=135,in=0]
        node[shift={(.25cm,.25cm)}] {\scriptsize $[{\tt g}=0]$}
        (a.east);
        \draw[->] (d.east) to [out=0,in=225]
        node[shift={(.25,-.25cm)}] {\scriptsize ${\tt g}:=1$}
        (first.south west);
        \draw[->] (a.south west) to [out=210, in=150]
        node[left] {\scriptsize $\uacc{\pi}$}
        (d.north west);
        \draw[->] (second.east) to [out=0, in=90]
        node[above right] {\scriptsize ${\tt g}:=2$}
        (more.north);
        \draw[->] (more.south) to [out=270, in=0]
        node[below right] {\scriptsize $[{\tt g}=2]$}
        (extra.east);
      \end{tikzpicture}
    \end{minipage}
  } 
  \caption{Accelerating a looping path\label{fig:instrument}}
\end{figure}

An example is shown in Figure~\ref{fig:instrument}.  The CFA
in Figure~\ref{fig:orig_cfa} represents an unaccelerated loop with a single path
through its body.  After adding an extra path to account for integer overflow,
we arrive at the CFA in Figure~\ref{fig:discr_cfa}.  We are able to find an accelerator
for the non-overflowing path, which we add to the CFA resulting in 
Figure~\ref{fig:acc_cfa}. 
We use $\uacc{\pi}$ to represent the accelerator $\pi$ for the
corresponding path. Then the restriction language is represented by
the regular expression $(\pi\,\vert\,\uacc{\pi}\cdot\uacc{\pi})$.
The corresponding 4-state trace automaton is shown in
Figure~\ref{fig:trace_aut}. By combining
the trace automaton and the CFA  we obtain the restricted CFA in 
Figure~\ref{fig:restrict_cfa} (after equivalent paths have been collapsed).

In the restricted CFA $\tilde{P}$, looping traces $\pi$
that can be accelerated and redundant iterations of
accelerators are infeasible and therefore do not trigger
the failure of unwinding assertions.
A CFA is safe if all unwinding assertions hold and 
no safety violation can be detected for a given bound~$k$. The reduction
of the diameter achieved by acceleration (Section~\ref{sec:reduction})
in combination with the construction presented in
this section enables us to establish the safety of CFAs
in cases in which traditional BMC would have been unable to do so.
Section~\ref{sec:experiments} provides
an experimental evaluation demonstrating the viability of
our approach.

%% file: experiments.tex
\section{Experimental Evaluation}
\label{sec:experiments}

We evaluate the effect of instrumenting accelerated programs
with trace automata and determine
the direct cost of constructing the automata as well as
the impact of trace automata on the ability to find bugs on the one hand
and prove safety on the other.

Our evaluation is based on the {\sc Loops} category of the benchmarks
from \SVCOMP and a number of small but difficult hand-crafted
examples.
Our hand-crafted examples require precise reasoning
about arithmetic and arrays.  The unsafe examples have deep bugs,
and the safe examples feature unbounded loops.
The \SVCOMP benchmarks are largely arithmetic in nature. They
often require non-trivial arithmetic invariants to be inferred, but rarely
require complex reasoning about arrays.  Furthermore, all bugs of
the unsafe \SVCOMP benchmarks occur within a small number of loop
iterations.

In all of our experiments we used \CBMC taken from the public SVN
at r3849 to perform the transformation. Since \CBMC's
acceleration procedure generates assertions with quantified arrays,
we used \textsc{Z3}~\cite{z3}
version 4.3.1 as the backend decision procedure.
All of the experiments were performed with a timeout of 30s and
very low unwinding limits.
We used an unwinding limit of 100 for unaccelerated programs
and an unwinding limit of 3 for their accelerated counterparts.

The version of \CBMC we use has incomplete acceleration
support, e.g., it is unable to accelerate nested loops.
As a result, there are numerous benchmarks that it cannot accelerate.
We stress that our goal here is to evaluate the effect of adding trace
automata to accelerated programs.  Acceleration has already proven to be a
useful technique for both bug-finding and
proof~\cite{KroeningLW13,DBLP:conf/sas/SchrammelJ11,DBLP:conf/pts/SchrammelMK13,philipp,KroeningW06}
and we are interested in how well inlined trace automata can complement it.

Our experimental results are summarised in Table~\ref{tbl:results-summary},
and the full results are shown in Appendix~\ref{sec:detailed_experiments}.
We discuss the results in the remainder of this section.

\begin{table}[t]\centering
\caption{Summary of experimental results\label{tbl:results-summary}}
\begin{tabular}{|l|c||c|r||c||c|r|r||c|r|r|}
\hhline{~----------}
\multicolumn{1}{c|}{~}
 & & \multicolumn{2}{|c||}{\CBMC} & &
\multicolumn{3}{c||}{\shortstack{\CBMC\\ +\\ Acceleration}} &
\multicolumn{3}{c|}{\shortstack{\CBMC + \\Acceleration + \\Trace
    Automata}} \\
\multicolumn{1}{c|}{~}
 & \rotatebox{90}{\#Benchmarks} & \rotatebox{90}{\#Correct} &
\multicolumn{1}{c||}{\rotatebox{90}{Time(s)}} &
\rotatebox{90}{\shortstack{\#Benchmarks \\[-.5ex] accelerated}}  &
\rotatebox{90}{\#Correct} & \rotatebox{90}{\shortstack{Acceleration
    \\[-.5ex] Time (s)}} & \rotatebox{90}{\shortstack{Checking \\ time
    (s)}}  & \rotatebox{90}{\#Correct} &
\rotatebox{90}{\shortstack{Acceleration \\ Time (s)}} &
\rotatebox{90}{\shortstack{Checking \\ Time (s)}} \\
 \hline
 \SVCOMP safe & 35 & 14 & 298.73 & 21 & 2 & 23.24 & 244.72 & 14 & 23.86 & 189.61 \\
 \SVCOMP unsafe & 32 & 20 & 394.96 & 18 & 11 & 15.79 & 197.94 & 12 & 16.51 & 173.74 \\
 Crafted safe & 15 & 0 & 11.42 & 15 & 0 & 2.75 & 32.41 & 15 & 2.91 & 1.59 \\
 Crafted unsafe & 14 & 0 & 9.03 & 14 & 14 & 2.85 & 12.24 & 14 & 2.95 & 2.55 \\
 \hline
\end{tabular}
\end{table}

\paragraph{Cost of Trace Automata.}

To evaluate the direct cost of constructing the trace automata, we
direct the reader's attention to Table~\ref{tbl:results-summary} and
the columns headed ``acceleration time''.
The first ``acceleration time'' column shows
how long it took to generate an accelerated program without a trace automaton,
whereas the second shows how long it took when a trace automaton was included.
For all of these benchmarks, the additional time taken to build and insert the trace
automaton is negligible.  The ``size increase'' column in 
Tables~\ref{tbl:results_safe}, \ref{tbl:results_unsafe}, and 
\ref{tbl:results_crafted} in Appendix~\ref{sec:detailed_experiments}
shows how much larger the instrumented binary is than the 
accelerated binary, expressed
as a percentage of the accelerated binary's size. The average
increase is about 15\%, but the maximum increase is 77\%.  There is
still room for optimisation, as we do not minimise the
automata before inserting them.

\paragraph{Bug Finding.}

In the following, 
we evaluate the effectiveness of our technique for bug finding.  The
current state-of-the-art method for bug finding is BMC~\cite{svcomp14}.
To provide a baseline for bug finding power, we
start by evaluating the effect of just combining acceleration with BMC.
We then evaluate the impact of adding trace automata, as
compared to acceleration without trace automata.  Our hypothesis is that
adding trace automata has negligible impact on acceleration's ability to
find bugs.  The statistics we use to measure these effects are the number of
bugs found and the time to find them.  We measure these statistics for each
of three techniques: BMC alone, acceleration with BMC, and our combination
of acceleration, trace automata and BMC.

The results are summarised in Table~\ref{tbl:results-summary}.
In \SVCOMP, almost all of the bugs occur after a small number of unwindings.
In these cases, there are no deep loops to accelerate so just using \CBMC allows
the same bugs to be reached, but without the overhead of acceleration (which
causes some timeouts to be hit).  In the
crafted set the bugs are much deeper, and we can see the effect of acceleration in
discovering these bugs -- none of the bugs are discovered by \CBMC, but each of the
configurations using acceleration finds all 14 bugs.

In both of the benchmark sets, adding trace automata does not negatively impact the bug finding ability
of acceleration.  Indeed, for the crafted set the addition of trace automata significantly
improves bug finding performance -- the total time needed to find the 14 bugs is reduced
from 12.31s to 1.85s.

\paragraph{Safety Proving.}

We evaluate the effectiveness of our technique for proving safety, the
key contribution of this paper.  Our two benchmark sets have very different
characteristics with respect to the safety proofs required for their safe
examples.  As can be seen from Table~\ref{tbl:results-summary}, 14 of the \SVCOMP
benchmarks can be proved safe using just BMC.  That is,
they can be exhaustively proved safe after a small number of loop unwindings.
For the 14 cases that were provable using just BMC, none
had loops that could execute for more than 10 iterations.

Of the 35 safe \SVCOMP benchmarks, 21 contained loops that could be accelerated.
Of these 21 cases, 14 were proved safe using trace automata.  These are not the
same 14 cases that were proved by \CBMC, and notably 8 cases with
unbounded loops are included, which would be impossible to prove safe with just BMC.
Additionally we were able to solve the \textsc{sum\_array\_true}
benchmark (shown in Fig.~\ref{fig:sumarray}) in 1.75s.  Of all the tools entered in
\SVCOMP, the only tools to claim ``safe'' for this benchmark were BMC-based,
and as such do not generate safety proofs.

For the 7 cases where accelerators were produced but we were unable to prove safety,
5 are due to timeouts, 1 is a crash in \CBMC and 1 is an ``incomplete''.  The 5 timeouts
are due to the complexity of the SMT queries we produce.  For these timeout cases, we
generate assertions which contain non-linear multiplication and quantification over arrays,
which are very difficult for Z3 to solve.  The ``incomplete'' case (\textsc{trex03\_true})
requires reasoning about accelerated paths that commute with each other, which we leave
as future work.

\begin{figure}[t]\centering
  \fbox{\begin{minipage}{.9\columnwidth}
      \begin{tabbing}
        \qquad\=\qquad\=\qquad\=\qquad\=\kill
        ${\tt unsigned~N:=*, i;}$\\
        ${\tt int~a[M], b[M], c[M]}$\\[1ex]
        ${\tt for~(i=0;\;i<M;\;i:=i+1)~\{}$\\
        \>${\tt c[i]:=a[i]+b[i];}$\\
        ${\tt\}}$\\[1ex]
        ${\tt for~(i=0;\;i<M;\;i:=i+1)~\{}$\\
        \>${\tt assert~(c[i]=a[i]+b[i]);}$\\
        ${\tt\}}$
  \end{tabbing}\end{minipage}}
  \caption{The \textsc{sum\_arrays} benchmark from \SVCOMP\label{fig:sumarray}}
\end{figure}

%% file: related.tex
\section{Related Work}
\label{sec:related}

The diameter of a transition system was introduced
in Biere et al.'s seminal paper on BMC~\cite{BiereCCZ99}
in the context of finite-state transition relations.
For finite-state transition relations, approximations
of the diameter can be computed 
symbolically by constraining the unwound transition relation
to exclude executions that visit states repeatedly~\cite{KroeningS03}.
For software, however, this technique is ineffective.
Baumgartner and K{\"u}hlmann use structural transformations of hardware designs
to reduce the reachability diameter of a hardware design to obtain a
complete BMC-based verification
method~\cite{DBLP:conf/date/BaumgartnerK04}.
This technique is not applicable in our context.

Trace automata are introduced in~\cite{HeizmannHP09}
as abstractions of safe traces of CFAs \cite{hrmgs02}, constructed
by means of interpolation. We use trace automata
to recognize redundant traces. 

Acceleration amounts to computing the transitive closure
of a infinite state transition relation~\cite{boigelot99,fl2002,BozgaIK10}.
Acceleration has been successfully combined with
abstract interpretation~\cite{DBLP:conf/sas/SchrammelJ11} as well
as interpolation-based invariant construction~\cite{philipp}. These 
techniques rely on over-approximate abstractions to prove safety.
We previously used acceleration and under-approximation to 
quickly find deep bugs~\cite{KroeningW06,kw2010,KroeningLW13}.
The quantified transition relations used to encode
under-appro\-xi\-ma\-tions pose an insurmountable challenge to 
interpolation-based refinement techniques~\cite{KroeningLW13},
making it difficult to combine the approach with traditional
software model checkers.

%% file: conclusion.tex
\section{Conclusion}
\label{sec:conclusion}

The reduction of the reachability diameter of
a program achieved by acceleration and
loop under-approximation enables the rapid detection
of bugs by means of BMC. 
Attempts to apply under-approximation to prove safety, however, 
have been disappointing:
the simple mechanism deployed by BMC-based tools to
detect that an unwinding bound is exhaustive is
not readily applicable to accelerated programs.

In this paper, we present a technique that constrains
the search space of an accelerated program, enabling
BMC-based tools to prove safety using a small unwinding
depth. To this end, we use \emph{trace automata} to 
eliminate redundant execution traces resulting from 
under-approximating acceleration.
Unlike other safety provers, our approach 
does not rely on over-approximation, nor does it require
the explicit computation of a fixed point. Using
unwinding assertions, the smaller diameter can be computed 
by means of a simple satisfiability check.

%% file: appendix_proofs.tex
\section{Proofs}
\label{sec:proofs}

\begin{lemma}
  Let $\uacc{\pi}$ be an under-approximating accelerator
  for the looping trace $\pi$. Then 
  $\uacc{\pi}\cdot\uacc{\pi}\preceq\uacc{\pi}$ holds.
\end{lemma}

\begin{proof}
  For accelerators that are not strictly
  under-approximating the claim holds trivially.
  Otherwise, we have
  \begin{displaymath}
    \begin{split}
    \langle\state,\state''\rangle\in\sem{\uacc{\pi}\cdot\uacc{\pi}}
    &\qquad\Leftrightarrow\\
    \exists \state'\,.\,\exists i,j\in\mathds{N}_0\,.&
    \left(\begin{array}{l}
      \langle\state,\state'\rangle\in\sem{\pi}^i\wedge
      i\leq\bound{\state}\quad\wedge\\
      \langle\state',\state''\rangle\in\sem{\pi}^j\;\wedge
      j\leq\bound{\state'}\\
    \end{array}\right)
    \end{split}
  \end{displaymath}
  If $\sigma'$ exists, 
  Condition \ref{eq:bound} in Definition~\ref{def:acc}
  guarantees that $\left(\bound{\state'}\leq\bound{\state}-i\right)$,
  and therefore  
  $\langle\state,\state''\rangle\in\sem{\uacc{\pi}\cdot\uacc{\pi}}$ 
  implies
  \begin{displaymath}
    \,\exists i,j\in\mathds{N}_0\,.
      \langle\state,\state''\rangle\in\sem{\pi}^{i+j}\wedge
      \underbrace{i\leq\bound{\state}\wedge j\leq\bound{\state}-i}_{
        (i+j)\leq\bound{\state}
      }\;.
  \end{displaymath}
  By replacing $i+j$ with a single variable $i$ we arrive
  at the definition of $\sem{\uacc{\pi}}$.
\end{proof}

\begin{lemma}
  Let $\acc{P}$ be an accelerated CFA, and
  $\lang_R$ be the corresponding restriction language.
  Let $\pi_1$ be a trace accepted by $\acc{P}$ such
  that $\pi_1\in\lang_R$. 
  Then there exists a trace $\pi_2$ which is accepted by $\acc{P}$
  such that $\pi_1\preceq\pi_2$ and $\pi_1$ is not a sub-trace
  of~$\pi_2$.
\end{lemma}

\begin{proof}
  The regular expression $(\pi\,\vert\,(\acc{\pi}\cdot\acc{\pi}))$
  can match the trace $\pi_1$ for two reasons:
  \begin{itemize}
  \item[(a)] The trace $\pi_1$ contains a sub-trace which 
    is a looping trace $\pi$ with a corresponding 
    accelerator $\acc{\pi}$ and $\pi\preceq\acc{\pi}$.
    We obtain $\pi_2$ by replacing 
    $\pi$ with $\acc{\pi}$. 
  \item[(b)] The trace $\pi_1$ contains the sub-trace
    $\acc{\pi}\cdot\acc{\pi}$ for some accelerator $\acc{\pi}$.
    Since $\acc{\pi}\cdot\acc{\pi}\preceq\acc{\pi}$
    (Lemma~\ref{lem:subsume_iteration}), we replace
    the sub-trace with $\acc{\pi}$ to obtain $\pi_2$.
  \end{itemize}
  Since the 
  accelerator $\acc{\pi}$ differs from the sub-trace it replaces
  in case (a), and $\vert\pi_2\vert<\vert\pi_1\vert$ in case (b),
  $\pi_1$ can not be contained in $\pi_2$. 
\end{proof}

%% file: results_table.tex
\section{Detailed Experimental Results}
\label{sec:detailed_experiments}

Tables~\ref{tbl:results_safe}, \ref{tbl:results_unsafe}, and 
\ref{tbl:results_crafted} show the detailed experimental results
for Table~\ref{tbl:results-summary} in Section~\ref{sec:experiments}.

\begin{sidewaystable}[t]\centering
{
\tiny
\begin{tabular}{|l|c||c|r||c||c|r|r||c|r|r|r|}
\hline
 & & \multicolumn{2}{|c||}{\CBMC} & Accelerated? & \multicolumn{3}{c||}{\CBMC + Acceleration} & \multicolumn{4}{c|}{\CBMC + Acceleration + Trace Automata} \\
 Name & Expected & Result & Time(s)               & & Result & \shortstack{Acceleration \\ time (s)} & \shortstack{Checking \\ time (s)}  & Result & \shortstack{Acceleration \\ time (s)} & \shortstack{Checking \\ time (s)} & \shortstack{Size \\ increase} \\
 \hline
 \hline
 \multicolumn{12}{|l|}{\SVCOMP} \\
 \hline
 \hline
 \input{svcomp_results_table_safe}
 \hline
\end{tabular}
}\\[2ex]

{\centering\scriptsize Key: Safe: \safe, Unsafe: \unsafe, Timeout: \timeout, Crash: \crash, Incomplete (unable to prove safety or find a bug): \incomplete}
\caption{Detailed experimental results for safe {\sc SVCOMP} benchmarks\label{tbl:results_safe}}
\end{sidewaystable}

\begin{sidewaystable}\centering
{
\tiny
\begin{tabular}{|l|c||c|r||c||c|r|r||c|r|r|r|}
\hline
 & & \multicolumn{2}{|c||}{\CBMC} & Accelerated? & \multicolumn{3}{c||}{\CBMC + Acceleration} & \multicolumn{4}{c|}{\CBMC + Acceleration + Trace Automata} \\
 Name & Expected & Result & Time(s)               & & Result & \shortstack{Acceleration \\ time (s)} & \shortstack{Checking \\ time (s)}  & Result & \shortstack{Acceleration \\ time (s)} & \shortstack{Checking \\ time (s)} & \shortstack{Size \\ increase} \\
 \hline
 \hline
 \multicolumn{12}{|l|}{\SVCOMP} \\
 \hline
 \hline
 \input{svcomp_results_table_unsafe}
 \hline
\end{tabular}
}\\[2ex]

{\centering\scriptsize Key: Safe: \safe, Unsafe: \unsafe, Timeout: \timeout, Crash: \crash, Incomplete (unable to prove safety or find a bug): \incomplete}
\caption{Detailed experimental results for unsafe {\sc SVCOMP} benchmarks\label{tbl:results_unsafe}}
\end{sidewaystable}

\begin{sidewaystable}[p]\centering
{
\tiny
\begin{tabular}{|l|c||c|r||c||c|r|r||c|r|r|r|}
\hline
 & & \multicolumn{2}{|c||}{\CBMC} & Accelerated? & \multicolumn{3}{c||}{\CBMC + Acceleration} & \multicolumn{4}{c|}{\CBMC + Acceleration + Trace Automata} \\
 Name & Expected & Result & Time(s)               & & Result & \shortstack{Acceleration \\ time (s)} & \shortstack{Checking \\ time (s)}  & Result & \shortstack{Acceleration \\ time (s)} & \shortstack{Checking \\ time (s)} & \shortstack{Size \\ increase} \\
 \hline
 \hline
 \multicolumn{12}{|l|}{Crafted} \\
 \hline
 \hline
 \input{crafted_results_table}
 \hline
\end{tabular}}\\[2ex]

{\centering\scriptsize Key: Safe: \safe, Unsafe: \unsafe, Timeout: \timeout, Crash: \crash, Incomplete (unable to prove safety or find a bug): \incomplete}
\caption{Detailed experimental results for crafted benchmarks\label{tbl:results_crafted}
}
\end{sidewaystable}

%% file: svcomp_results_table_safe.tex
array\_true.c & \safe & \safe & 0.04s & & --- & --- & --- & --- & --- & --- & --- \\
bubble\_sort\_true.c & \safe & \timeout & 30.00s & Yes & \timeout & 3.90s & 30.00s &\timeout & 3.97s & 30.00s &20\% \\
count\_up\_down\_true.c & \safe & \incomplete & 0.84s & Yes & \incomplete & 0.20s & 1.20s &\safe & 0.21s & 0.25s &11\% \\
eureka\_01\_true.c & \safe & \safe & 12.64s & Yes & \timeout & 1.90s & 30.00s &\timeout & 1.95s & 30.00s &34\% \\
eureka\_05\_true.c & \safe & \safe & 0.11s & Yes & \crash & 0.56s & 2.64s &\safe & 0.57s & 2.22s &31\% \\
for\_infinite\_loop\_1\_true.c & \safe & \incomplete & 0.05s & Yes & \incomplete & 0.12s & 0.09s &\safe & 0.13s & 0.06s &11\% \\
for\_infinite\_loop\_2\_true.c & \safe & \unsafe & 0.08s & Yes & \unsafe & 0.13s & 0.05s &\safe & 0.14s & 0.07s &12\% \\
heavy\_true.c & \safe & \timeout & 30.00s & & --- & --- & --- & --- & --- & --- & --- \\
insertion\_sort\_true.c & \safe & \timeout & 30.00s & Yes & \crash & 0.46s & 30.00s &\crash & 0.48s & 30.00s &18\% \\
invert\_string\_true.c & \safe & \safe & 0.12s & Yes & \crash & 0.87s & 30.00s &\safe & 0.92s & 2.13s &28\% \\
linear\_sea.ch\_true.c & \safe & \incomplete & 3.78s & Yes & \timeout & 0.33s & 30.00s &\safe & 0.35s & 0.26s &20\% \\
lu.cmp\_true.c & \safe & \safe & 0.34s & & --- & --- & --- & --- & --- & --- & --- \\
matrix\_true.c & \safe & \safe & 0.03s & & --- & --- & --- & --- & --- & --- & --- \\
n.c11\_true.c & \safe & \incomplete & 0.91s & & --- & --- & --- & --- & --- & --- & --- \\
n.c24\_true.c & \safe & \timeout & 30.00s & Yes & \crash & 3.60s & 11.41s &\timeout & 3.66s & 30.00s &17\% \\
n.c40\_true.c & \safe & \safe & 0.04s & Yes & \safe & 0.25s & 0.14s &\safe & 0.26s & 0.15s &11\% \\
nec40\_true.c & \safe & \safe & 0.04s & Yes & \safe & 0.25s & 0.13s &\safe & 0.25s & 0.17s &11\% \\
string\_true.c & \safe & \safe & 11.20s & & --- & --- & --- & --- & --- & --- & --- \\
sum01\_true.c & \safe & \incomplete & 0.81s & Yes & \incomplete & 0.50s & 6.24s &\safe & 0.51s & 0.43s &19\% \\
sum03\_true.c & \safe & \incomplete & 0.07s & Yes & \incomplete & 0.47s & 0.23s &\safe & 0.46s & 0.22s &17\% \\
sum04\_true.c & \safe & \safe & 0.00s & Yes & \incomplete & 0.23s & 0.22s &\safe & 0.24s & 0.13s &11\% \\
sum\_array\_true.c & \safe & \crash & 30.00s & Yes & \crash & 0.56s & 30.00s &\safe & 0.62s & 1.75s &29\% \\
terminator\_02\_true.c & \safe & \safe & 2.58s & & --- & --- & --- & --- & --- & --- & --- \\
terminator\_03\_true.c & \safe & \timeout & 30.00s & & --- & --- & --- & --- & --- & --- & --- \\
trex01\_true.c & \safe & \incomplete & 13.96s & & --- & --- & --- & --- & --- & --- & --- \\
trex02\_true.c & \safe & \incomplete & 1.27s & & --- & --- & --- & --- & --- & --- & --- \\
trex03\_true.c & \safe & \incomplete & 9.51s & Yes & \incomplete & 6.22s & 0.75s &\incomplete & 6.09s & 1.69s &54\% \\
trex04\_true.c & \safe & \incomplete & 0.91s & & --- & --- & --- & --- & --- & --- & --- \\
veris.c\_NetBSD-libc\_\_loop\_true.c & \safe & \safe & 17.61s & & --- & --- & --- & --- & --- & --- & --- \\
veris.c\_OpenSER\_\_cases1\_stripFullBoth\_arr\_true.c & \safe & \timeout & 30.00s & Yes & \incomplete & 1.05s & 11.58s &\timeout & 1.16s & 30.00s &77\% \\
veris.c\_sendmail\_\_tTflag\_arr\_one\_loop\_true.c & \safe & \safe & 0.88s & & --- & --- & --- & --- & --- & --- & --- \\
vogal\_true.c & \safe & \safe & 10.75s & Yes & \timeout & 1.60s & 30.00s &\timeout & 1.85s & 30.00s &64\% \\
while\_infinite\_loop\_1\_true.c & \safe & \incomplete & 0.03s & Yes & \incomplete & 0.01s & 0.02s &\safe & 0.01s & 0.03s &15\% \\
while\_infinite\_loop\_2\_true.c & \safe & \incomplete & 0.06s & Yes & \incomplete & 0.03s & 0.02s &\safe & 0.03s & 0.05s &16\% \\
while\_infinite\_loop\_3\_true.c & \safe & \incomplete & 0.07s & & --- & --- & --- & --- & --- & --- & --- \\
\hline
Total & 35 & 14 & 298.73s & 21 & 2 & 23.24s & 244.72s & 14 & 23.86s & 189.61s & \\
\hline

%% file: svcomp_results_table_unsafe.tex
array\_false.c & \unsafe & \unsafe & 0.03s & & --- & --- & --- & --- & --- & --- & --- \\
bubble\_sort\_false.c & \unsafe & \timeout & 30.00s & & --- & --- & --- & --- & --- & --- & --- \\
compact\_false.c & \unsafe & \timeout & 30.00s & & --- & --- & --- & --- & --- & --- & --- \\
count\_up\_down\_false.c & \unsafe & \unsafe & 0.26s & Yes & \unsafe & 0.20s & 0.22s &\unsafe & 0.21s & 0.30s &11\% \\
eureka\_01\_false.c & \unsafe & \timeout & 30.00s & Yes & \timeout & 1.98s & 30.00s &\timeout & 2.00s & 30.00s &27\% \\
for\_bounded\_loop1\_false.c & \unsafe & \unsafe & 0.67s & & --- & --- & --- & --- & --- & --- & --- \\
heavy\_false.c & \unsafe & \timeout & 30.00s & & --- & --- & --- & --- & --- & --- & --- \\
insertion\_sort\_false.c & \unsafe & \timeout & 30.00s & Yes & \crash & 0.64s & 12.31s &\crash & 0.62s & 14.58s &17\% \\
invert\_string\_false.c & \unsafe & \timeout & 30.00s & Yes & \crash & 0.61s & 30.00s &\unsafe & 0.64s & 3.00s &17\% \\
linear\_search\_false.c & \unsafe & \unsafe & 0.47s & Yes & \unsafe & 0.36s & 0.18s &\unsafe & 0.38s & 0.34s &20\% \\
ludcmp\_false.c & \unsafe & \unsafe & 0.45s & & --- & --- & --- & --- & --- & --- & --- \\
matrix\_false.c & \unsafe & \timeout & 30.00s & Yes & \timeout & 0.24s & 30.00s &\timeout & 0.28s & 30.00s &19\% \\
nec11\_false.c & \unsafe & \unsafe & 0.29s & Yes & \unsafe & 0.13s & 0.08s &\unsafe & 0.14s & 0.12s &12\% \\
nec20\_false.c & \unsafe & \unsafe & 0.24s & Yes & \unsafe & 0.51s & 0.37s &\unsafe & 0.52s & 0.44s &17\% \\
string\_false.c & \unsafe & \crash & 30.00s & & --- & --- & --- & --- & --- & --- & --- \\
sum01\_bug02\_false.c & \unsafe & \unsafe & 0.27s & Yes & \unsafe & 1.89s & 0.82s &\unsafe & 1.95s & 1.01s &27\% \\
sum01\_bug02\_sum01\_bug02\_base.case\_false.c & \unsafe & \unsafe & 0.26s & Yes & \unsafe & 0.45s & 1.94s &\unsafe & 0.47s & 0.80s &20\% \\
sum01\_false.c & \unsafe & \unsafe & 0.22s & Yes & \unsafe & 0.46s & 0.35s &\unsafe & 0.47s & 0.30s &22\% \\
sum03\_false.c & \unsafe & \unsafe & 2.45s & Yes & \unsafe & 0.65s & 0.65s &\unsafe & 0.70s & 0.80s &32\% \\
sum04\_false.c & \unsafe & \unsafe & 0.05s & Yes & \unsafe & 0.35s & 0.19s &\unsafe & 0.36s & 0.25s &24\% \\
sum\_array\_false.c & \unsafe & \crash & 30.00s & Yes & \timeout & 0.59s & 30.00s &\timeout & 0.64s & 30.00s &28\% \\
terminator\_01\_false.c & \unsafe & \unsafe & 0.28s & Yes & \unsafe & 0.12s & 0.13s &\unsafe & 0.12s & 0.15s &12\% \\
terminator\_02\_false.c & \unsafe & \unsafe & 3.56s & & --- & --- & --- & --- & --- & --- & --- \\
terminator\_03\_false.c & \unsafe & \unsafe & 0.42s & & --- & --- & --- & --- & --- & --- & --- \\
trex01\_false.c & \unsafe & \unsafe & 2.69s & & --- & --- & --- & --- & --- & --- & --- \\
trex02\_false.c & \unsafe & \unsafe & 0.66s & & --- & --- & --- & --- & --- & --- & --- \\
trex03\_false.c & \unsafe & \unsafe & 8.21s & Yes & \unsafe & 3.96s & 0.70s &\unsafe & 3.98s & 1.65s &54\% \\
verisec\_NetBSD-libc\_\_loop\_false.c & \unsafe & \unsafe & 10.45s & & --- & --- & --- & --- & --- & --- & --- \\
verisec\_OpenSER\_\_cases1\_stripFullBoth\_arr\_false.c & \unsafe & \timeout & 30.00s & Yes & \timeout & 1.03s & 30.00s &\timeout & 1.20s & 30.00s &76\% \\
verisec\_sendmail\_\_tTflag\_arr\_one\_loop\_false.c & \unsafe & \timeout & 30.00s & & --- & --- & --- & --- & --- & --- & --- \\
vogal\_false.c & \unsafe & \crash & 30.00s & Yes & \timeout & 1.62s & 30.00s &\timeout & 1.83s & 30.00s &68\% \\
while\_infinite\_loop\_4\_false.c & \unsafe & \unsafe & 3.03s & & --- & --- & --- & --- & --- & --- & --- \\
\hline
Total & 32 & 20 & 394.96s & 18 & 11 & 15.79s & 197.94s & 12 & 16.51s & 173.74s & \\
\hline

%% file: crafted_results_table.tex
array\_safe1 & \safe & \incomplete & 0.20s & Yes & \incomplete & 0.15s & 0.27s &\safe & 0.16s & 0.08s &11\% \\
array\_safe2 & \safe & \incomplete & 0.08s & Yes & \timeout & 0.14s & 30.00s &\safe & 0.13s & 0.09s &10\% \\
array\_safe3 & \safe & \incomplete & 0.48s & Yes & \incomplete & 0.12s & 0.28s &\safe & 0.14s & 0.07s &15\% \\
array\_safe4 & \safe & \incomplete & 0.49s & Yes & \incomplete & 0.12s & 0.19s &\safe & 0.14s & 0.05s &13\% \\
const\_safe1 & \safe & \incomplete & 0.27s & Yes & \incomplete & 0.15s & 0.06s &\safe & 0.15s & 0.08s &12\% \\
diamond\_safe1 & \safe & \incomplete & 0.51s & Yes & \incomplete & 0.18s & 0.15s &\safe & 0.18s & 0.13s &26\% \\
diamond\_safe2 & \safe & \incomplete & 7.53s & Yes & \incomplete & 0.66s & 0.77s &\safe & 0.66s & 0.45s &32\% \\
functions\_safe1 & \safe & \incomplete & 0.06s & Yes & \incomplete & 0.13s & 0.08s &\safe & 0.15s & 0.06s &12\% \\
multivar\_safe1 & \safe & \incomplete & 0.40s & Yes & \incomplete & 0.18s & 0.08s &\safe & 0.19s & 0.07s &12\% \\
overflow\_safe1 & \safe & \incomplete & 0.04s & Yes & \incomplete & 0.13s & 0.06s &\safe & 0.14s & 0.07s &13\% \\
phases\_safe1 & \safe & \incomplete & 0.04s & Yes & \incomplete & 0.23s & 0.09s &\safe & 0.25s & 0.14s &26\% \\
simple\_safe1 & \safe & \incomplete & 0.04s & Yes & \incomplete & 0.14s & 0.08s &\safe & 0.15s & 0.06s &13\% \\
simple\_safe2 & \safe & \incomplete & 1.00s & Yes & \incomplete & 0.13s & 0.09s &\safe & 0.15s & 0.10s &13\% \\
simple\_safe3 & \safe & \incomplete & 0.24s & Yes & \incomplete & 0.13s & 0.07s &\safe & 0.14s & 0.07s &13\% \\
simple\_safe4 & \safe & \incomplete & 0.04s & Yes & \incomplete & 0.16s & 0.14s &\safe & 0.18s & 0.07s &13\% \\
\hline
Total & 15 & 0 & 11.42s & 15 & 0 & 2.75s & 32.41s & 15 & 2.91s & 1.59s & \\
\hline
array\_unsafe1 & \unsafe & \incomplete & 0.49s & Yes & \unsafe & 0.12s & 0.04s &\unsafe & 0.13s & 0.06s &14\% \\
array\_unsafe2 & \unsafe & \incomplete & 0.09s & Yes & \unsafe & 0.15s & 10.42s &\unsafe & 0.15s & 0.11s &10\% \\
array\_unsafe3 & \unsafe & \incomplete & 0.48s & Yes & \unsafe & 0.14s & 0.04s &\unsafe & 0.14s & 0.06s &14\% \\
const\_unsafe1 & \unsafe & \incomplete & 0.05s & Yes & \unsafe & 0.12s & 0.05s &\unsafe & 0.13s & 0.07s &12\% \\
diamond\_unsafe1 & \unsafe & \incomplete & 0.51s & Yes & \unsafe & 0.25s & 0.10s &\unsafe & 0.24s & 0.20s &26\% \\
diamond\_unsafe2 & \unsafe & \incomplete & 7.03s & Yes & \unsafe & 0.86s & 0.88s &\unsafe & 0.89s & 1.41s &33\% \\
functions\_unsafe1 & \unsafe & \incomplete & 0.06s & Yes & \unsafe & 0.13s & 0.07s &\unsafe & 0.12s & 0.07s &12\% \\
multivar\_unsafe1 & \unsafe & \incomplete & 0.05s & Yes & \unsafe & 0.20s & 0.12s &\unsafe & 0.19s & 0.08s &11\% \\
overflow\_unsafe1 & \unsafe & \incomplete & 0.05s & Yes & \unsafe & 0.13s & 0.09s &\unsafe & 0.16s & 0.08s &13\% \\
phases\_unsafe1 & \unsafe & \incomplete & 0.06s & Yes & \unsafe & 0.23s & 0.10s &\unsafe & 0.23s & 0.13s &26\% \\
simple\_unsafe1 & \unsafe & \incomplete & 0.04s & Yes & \unsafe & 0.12s & 0.06s &\unsafe & 0.15s & 0.06s &13\% \\
simple\_unsafe2 & \unsafe & \incomplete & 0.04s & Yes & \unsafe & 0.12s & 0.05s &\unsafe & 0.13s & 0.06s &13\% \\
simple\_unsafe3 & \unsafe & \incomplete & 0.04s & Yes & \unsafe & 0.13s & 0.06s &\unsafe & 0.14s & 0.07s &12\% \\
simple\_unsafe4 & \unsafe & \incomplete & 0.04s & Yes & \unsafe & 0.15s & 0.16s &\unsafe & 0.15s & 0.09s &13\% \\
\hline
Total & 14 & 0 & 9.03s & 14 & 14 & 2.85s & 12.24s & 14 & 2.95s & 2.55s & \\
\hline